\documentclass[a4paper,11pt]{article}
\pdfoutput=1

\usepackage[a4paper,tmargin=3truecm,bmargin=3truecm,rmargin=2.5truecm,
lmargin=2.5truecm,twoside,verbose=true]{geometry}
\usepackage{cancel,graphicx}

\usepackage{amsmath,amssymb}
\usepackage[amsmath, hyperref, thmmarks]{ntheorem}
\usepackage[all]{xypic}
\usepackage[pdftex]{hyperref}
\usepackage[english]{babel}

\numberwithin{equation}{section}

\allowdisplaybreaks[1]

\newcommand\caA{{\mathcal A}}

\newcommand\caI{{\mathcal I}}
\newcommand\caL{{\mathcal L}}
\newcommand\caT{{\mathcal T}}

\newcommand\caS{{\mathcal S}}

\newcommand\caV{{\mathcal V}}

\newcommand\wx{{\widetilde x}}

\newcommand\gR{{\mathbb R}}

\newcommand\gN{{\mathbb N}}

\newcommand\eps{{\varepsilon}}

\newcommand\dd{{\text{\textup{d}}}}

\newcommand\norm{\mathord{\parallel}}

\theoremsymbol{}
\theorembodyfont{\slshape}
\theoremheaderfont{\normalfont\bfseries}
\theoremseparator{}
\newtheorem{Theorem}{Theorem}[section]
\newtheorem{theorem}[Theorem]{Theorem}

\newtheorem{Lemma}[Theorem]{Lemma}

\theorembodyfont{\upshape}
\theoremsymbol{\ensuremath{\blacklozenge}}

\newtheorem{definition}[Theorem]{Definition}
\newtheorem*{feynman}[Theorem]{Feynman rules:}

\theoremstyle{nonumberplain}
\theoremheaderfont{\scshape}
\theorembodyfont{\normalfont}
\theoremsymbol{\ensuremath{\blacksquare}}

\newtheorem{proof}{Proof}
\qedsymbol{\ensuremath{_\blacksquare}}
\theoremclass{LaTeX}

\title{Renormalization of the commutative scalar theory with harmonic term to all orders}
\author{Axel de Goursac}
\date{}

\begin{document}

\maketitle
\vspace*{-1cm}
\begin{center}
\textit{Charg\'e de Recherche au F.R.S.-FNRS\\ IRMP, Universit\'e Catholique de Louvain,\\ Chemin du Cyclotron, 2, B-1348 Louvain-la-Neuve, Belgium,}\\
\textit{and Max Planck Institut f\"ur Mathematik,\\ Vivatsgasse 7, D-53111 Bonn, Germany\\
  e-mail: \texttt{axelmg@melix.net}}\\
\end{center}%

\vskip 2cm

\begin{abstract}
The noncommutative scalar theory with harmonic term (on the Moyal space) has a vanishing beta function. In this paper, we prove the renormalizability of the commutative scalar field theory with harmonic term to all orders by using multiscale analysis in the momentum space. Then, we consider and compute its one-loop beta function, as well as the one on the degenerate Moyal space. We can finally compare both to the vanishing beta function of the theory with harmonic term on the Moyal space.
\end{abstract}

\vfill

\pagebreak

\section{Introduction}

Noncommutative quantum field theories appear to be interesting candidates for new Physics beyond the Standard Model of particles (see \cite{Douglas:2001ba,Wulkenhaar:2006si} for a review). They are based on ``spaces'' construced from Noncommutative Geometry \cite{Connes:1994} and involve new features concerning renormalizability, vacuum configurations,... The recent discovery at the LHC of a new particle, which could be the Higgs boson, reinforces the interest for scalar theories in general, and especially the interpretation of the Higgs field as part of a noncommutative connection (i.e. a noncommutative gauge potential): see \cite{DuboisViolette:1988ir,Connes:1990lo} for almost-commutative geometries and \cite{deGoursac:2008bd,Cagnache:2008tz} for the Moyal geometry.
\medskip

In 2004, the $\phi^4$ scalar field theory on the (noncommutative) Moyal space was showed to be renormalizable at all orders in perturbation \cite{Grosse:2004yu} only if a harmonic term $\int x^2\phi^2(x)$ was added to the action. Otherwise, a new divergence, called UV-IR mixing \cite{Minwalla:1999px}, spoils the renormalizability. Then, numerous proofs \cite{Rivasseau:2005bh,Gurau:2005gd,Gurau:2007fy,Tanasa:2007xa} of the renormalizability of this action as well as studies of the renormalization flow \cite{Grosse:2004by,Disertori:2006nq} showed that this theory is asymptotically safe and therefore involves a new renormalization group. Even more, it seems to be a non-perturbatively solvable model \cite{Grosse:2012uv}. This theory has other remarkable properties, like a new symmetry called Langmann-Szabo duality \cite{Langmann:2002cc} and interpreted as a grading symmetry \cite{deGoursac:2010zb}. The vacuum solutions of the theory with a negative mass term ($m^2<0$) have been exhibited in \cite{deGoursac:2007uv} and are not constant. The noncommutative Noether currents have been computed in \cite{Hounkonnou:2009qt,Geloun:2007zz}. Moreover, even if the choice of a deformation structure $\Theta$ breaks the rotational invariance of the model, it can be restored at all orders in the renormalization procedure \cite{deGoursac:2009fm}. Finally, let us mention that the commutative limit of this theory is not well-defined, and this will be partly discussed in the present paper.

Note that there are now other renormalizable theories on the Moyal space. For instance, the complex scalar LSZ-model \cite{Langmann:2003if} and the Gross-Neveu fermionic model \cite{VignesTourneret:2006nb}. Another renormalizable real scalar model on the Moyal space has been exhibited \cite{Gurau:2008vd}, in which the non-local IR counterterm $\int\frac{1}{p^2}\hat\phi(p)\hat\phi(-p)$ is now included in the classical action. The resulting theory is translation-invariant, but does not possess the properties exposed above for the Grosse-Wulkenhaar model (see also \cite{Geloun:2008hr,Magnen:2008pd,Blaschke:2008yj,Blaschke:2009gm}). Some new features can also be found in other noncommutative scalar theories \cite{Pinzul:2011di,Liang:2010zza}. A noncommutative gauge theory involving a harmonic term has been constructed in \cite{deGoursac:2007gq,Grosse:2007dm}, which is therefore strongly related to the Grosse-Wulkenhaar model. It admits only non-trivial vacuum configurations \cite{deGoursac:2008rb} and is a good candidate to renormalizability. The BRST framework of this gauge model has been investigated in \cite{Blaschke:2007vc,Blaschke:2009aw}.
\medskip

Besides, a recent paper \cite{Wulkenhaar:2009pv} has showed that the $\phi^4$ scalar field theory with harmonic term but on the usual commutative space $\gR^4$ emerges from the spectral action of a supersymmetric spectral triple. The investigation of its (non-constant) vacuum configurations led to a new interpretation of the Higgs mechanism which gives mass to the fermionic and gauge fields. The phase transition of the spontaneous symmetry breaking mechanism depends on the configuration coordinate $x$ in this theory, as well as on the mass flow. Nonetheless, the question of renormalizability of the theory was not taken into account in this analysis.
\medskip

The aim of this paper is first to renormalize to all orders the commutative $\phi^4$ scalar field theory with harmonic term, whose interest has been underlined just above, by using multiscale analysis in the momentum space as it has been done for noncommutative field theories but in the position space \cite{Rivasseau:2005bh,Gurau:2005gd}. Then, it makes sense to compute the (one-loop) beta function and to compare it with the one of the Moyal space \cite{Grosse:2004by} and the one of the degenerate Moyal space. We want indeed to describe how the commutative limit is pathologic at the level of beta functions, so that the asymptotic safety is not preserved.


\section{Renormalization of the commutative field theory with harmonic term}
\label{sec-ren}

\subsection{Presentation of the theory}

The scalar theory with harmonic term on the (Euclidean) commutative space $\gR^D$ is given by:
\begin{equation}
S(\phi)=\int \dd^Dx\Big(\frac 12(\partial_\mu\phi)^2 +\frac{\Omega^2}{2}x^2\phi^2 +\frac{m^2}{2}\phi^2 +\lambda\phi^4\Big),\label{eq-com-actharm}
\end{equation}
where $\Omega$ is a real parameter of mass dimension 2, and the commutative pointwise product has been used. The propagator of this theory is given by:
\begin{align}
C(x,y)&=\left(\frac{\Omega}{2\pi}\right)^{\frac D2}\int_0^\infty \!\! \dd t\ C(t,x,y),\nonumber\\
C(t,x,y) &= \frac{e^{-m^2t}}{\sinh^{\frac D2}(2\Omega t)}\exp\Big(-\frac{\Omega}{4}\coth(\Omega t)(x-y)^2 -\frac{\Omega}{4}\tanh(\Omega t)(x+y)^2\Big).\label{eq-propag3}
\end{align}
Its Fourier transform is given by
\begin{align}
\hat C(p,q)&=\left(\frac{2\pi}{\Omega}\right)^{\frac D2}\int_0^\infty \!\! \dd t\ \hat C(t,p,q),\nonumber\\
\hat C(t,p,q) &= \frac{e^{-m^2t}}{\sinh^{\frac D2}(2\Omega t)} \exp\Big(-\frac{1}{4\Omega}\coth(\Omega t)(p+q)^2 -\frac{1}{4\Omega}\tanh(\Omega t)(p-q)^2\Big),\label{eq-propag4}
\end{align}
where we use the convention $\hat f(p)=\int\dd^Dx\ f(x)e^{-ipx}$ in this paper.

The mass dimension of the parameter $\Omega$ is an indication of the renormalizability of this theory, but it does not prove wether this theory is stable under renormalization or not (see section \ref{sec-disc}). That is why we study its renormalizability with multiscale analysis in the following.

\subsection{Power-counting}

We use below the multiscale analysis \cite{Rivasseau:1991,VignesTourneret:2006xa}. For that, we consider the regularization of the Fourier transform of the propagator:
\begin{equation*}
\hat C_\rho(p,q)=\left(\frac{2\pi}{\Omega}\right)^{\frac D2}\int_{M^{-2\rho}}^\infty \!\! \dd t\ \hat C(t,p,q)
\end{equation*}
where $M> 1$ is a fixed number in the following and $\rho\in\gN$ plays the role of an ultraviolet cut-off. Then, we cut this regularized propagator into slices:
\begin{align}
&\hat C_\rho(p,q)=\sum_{i=0}^\rho\hat C^i(p,q),\label{eq-slice}\\
&\hat C^0(p,q)=\left(\frac{2\pi}{\Omega}\right)^{\frac D2}\int_{1}^\infty \!\! \dd t\ \hat C(t,p,q),\qquad \hat C^i(p,q)=\left(\frac{2\pi}{\Omega}\right)^{\frac D2}\int_{M^{-2i}}^{M^{-2(i-1)}} \!\! \dd t\ \hat C(t,p,q)\nonumber
\end{align}
for $i\in\{1,\dots,\rho\}$. $i$ is called the {\bf scale} or the index of the propagator $\hat C^i$.
\medskip

\begin{feynman}
Let $G$ be an amputated Feynman graph of this theory. We denote by $\caI(G)$ the set of internal lines of $G$, $I(G)$ the cardinal of $\caI(G)$, $\caV(G)$ the set of vertices of $G$, $n(G)$ the cardinal of $\caV(G)$, $N(G)>0$ its number of external legs and $L(G)$ its number of loops. The parity of the theory implies that only the amplitudes of graphs for even $N$ are non-vanishing. Due to the Euler characteristic and to the $\phi^4$ theory, we also have the classical identities:
\begin{equation}
L(G)=I(G)-n(G)+1,\qquad 4n(G)=2I(G)+N(G)\label{eq-euler}
\end{equation}
where we assume that $G$ is connected in the first identity. Finally, for a vertex $\nu\in\caV(G)$, we define $\caL_\nu(G)$ to be the set of internal or external lines $\ell$ of $G$ hooked to $\nu$. The Feynman rules are as follows. For each internal line $\ell$ with incoming impulsions $p_\ell$ and $q_\ell$ at each boundary, $(2\pi)^D\int\dd^Dp_\ell\dd^D q_\ell\ \hat C_\rho(p_\ell,q_\ell)$ contributes, while the contribution for each vertex $\nu$ is $\frac{\lambda}{(2\pi)^{3D}}\delta(\sum_{\ell\in\caL_\nu(G)}p_\ell)$ where the $p_\ell$'s are the four incoming (internal or external) impulsions to this vertex $\nu$.
\end{feynman}
Due to these Feynman rules, we can write the regularized amplitude $A_G$ of the graph $G$:
\begin{equation}
A_G=\frac{\lambda^{n(G)}}{(2\pi)^{D(3n(G)-I(G))}}\int \prod_{\ell\in\caI(G)}\dd^Dp_\ell\dd^Dq_\ell\ \hat C_\rho(p_\ell,q_\ell)\ \Delta(\{p_\ell,q_\ell,k_e\})\label{eq-ampl}
\end{equation}
where the distribution $\Delta$ depends on all (internal and external) impulsions and summarize all the delta functions of each vertex. Note that $A_G$ is a distribution on the external impulsions $k_e$ where we take the convention in the following that the $k_e$'s are incoming ($e\in\{1,\dots,N(G)\}$). By using the decomposition of the propagator into slices \eqref{eq-slice}, the amplitude will be a sum over the scales $i_\ell$ of every internal lines $\ell$. We call $\mu=\{i_\ell\}_{\ell\in\caI(G)}\in\{0,\dots,\rho\}^{I(G)}$ an {\bf attribution} of the graph $G$. Then
\begin{equation}
A_G=\frac{\lambda^{n(G)}}{(2\pi)^{D(3n(G)-I(G))}}\sum_\mu A_G^\mu,\quad\text{ with } A_G^\mu=\int \prod_{\ell\in\caI(G)}\dd^Dp_\ell\dd^Dq_\ell\ \hat C^{i_\ell}(p_\ell,q_\ell)\ \Delta(\{p_\ell,q_\ell,k_e\}).\label{eq-amplattr}
\end{equation}
Now we introduce a key concept in multiscale analysis.
\begin{definition}
\label{def-quasilocal}
For a given attribution $\mu$ and $i\in\{0,\dots,\rho\}$, we consider the subgraph of $G$ composed of the lines $\ell\in\caI(G)$ with index $i_\ell\geq i$ and vertices hooked to these lines. It is not connected in general. We denote by $G^i_k$ the connected components (indexed by $k$) of this subgraph, and call them the {\bf quasilocal subgraphs} of $G$. By convention, $G$ is also a quasilocal subgraph and the indices of its external legs are set to -1. This name is justified by the fact that all indices of internal lines of a $G^i_k$ are greater or equal to $i$ while all indices of external lines of $G^i_k$ (but internal lines in $G$) are smaller than $i$, so $G^i_k$ has a smaller spatial extension than its external legs.
\end{definition}

We first want to have a bound for the propagator in a slice.
\begin{Lemma}
\label{lem-prop}
The Fourier transform of the propagator admits the following bound: there exists constants $K,k,k'>0$ such that for any $i\in\{0,\dots,\rho\}$, $\forall p,q\in\gR^D$,
\begin{equation*}
|\hat C^i(p,q)|\leq K M^{(i+1)(D-2)}e^{-kM^{2(i+1)}(p+q)^2-k'M^{-2(i+1)}(p-q)^2}
\end{equation*}
where $k=\frac{1}{4\Omega}$.
\end{Lemma}
\begin{proof}
For $i=0$, one integrates $t$ over $[1,+\infty)$. One easily obtains the bounds $\sinh(2\Omega t)\geq \sinh(2\Omega)$, $\exp\Big(-\frac{1}{4\Omega}\coth(\Omega t)(p+q)^2\Big)\leq \exp\Big(-\frac{1}{4\Omega}(p+q)^2\Big)$ and $\exp\Big(-\frac{1}{4\Omega}\tanh(\Omega t)(p-q)^2\Big)\leq \exp\Big(-\frac{1}{4\Omega}\tanh(\Omega)(p-q)^2\Big)$, which gives the result.

For $i\geq 1$, we have to integrate $t$ over $[M^{-2i},M^{-2(i-1)}]$. By using Taylor inequality, there exists constants $\alpha,\beta,\gamma>0$ independent of $i$ such that $\forall t\in]0,1]$, $\sinh(t)\geq\alpha t$, $\coth(t)\geq \frac{1}{\beta t}$ and $\tanh(t)\geq\gamma t$. Inserting these inequalities in \eqref{eq-propag4} produces:
\begin{equation*}
|\hat C^i(p,q)|\leq \left(\frac{2\pi}{\Omega}\right)^{\frac D2} (M^{-2(i-1)}-M^{-2i})\frac{M^{Di}}{(2\alpha\Omega)^{\frac D2}} \exp\Big(-\frac{M^{2(i-1)}}{4\beta\Omega^2}(p+q)^2 -\frac{M^{-2i}}{4\gamma}(p-q)^2\Big)
\end{equation*}
We set $K=\left(\frac{\pi}{\alpha\Omega^2}\right)^{\frac D2}(M^2-1)M^{2-D}$, $k=\frac{1}{4\beta\Omega^2 M^4}$, $k'=\frac{M^2}{4\gamma}$, and obtain the result.
\end{proof}

\begin{theorem}[Power counting]
\label{thm-powcount}
Let $G$ be a connected amputated graph of the theory \eqref{eq-com-actharm}. Then there exists a constant $K>0$ such that for any attribution $\mu$ of $G$, $\forall \varphi_j\in\caS(\gR^D)$,
\begin{multline*}
\int\dd^Dk_1\dots\dd^Dk_N\ A_G^\mu(k_1,\dots,k_N)\hat\varphi_1(k_1)\dots\hat\varphi_N(k_N)\leq\\
K\norm\hat\varphi_1\norm_1\dots \norm\hat\varphi_{N-2}\norm_1\norm\hat\varphi_{N-1}\norm_2\norm\hat\varphi_N\norm_2\prod_{(i,k)}M^{-\omega(G^i_k)},
\end{multline*}
where $G^i_k$ are the quasilocal subgraphs of $G$ and the {\bf superficial degree of divergence} is given by
\begin{equation}
\omega(G)=(4-D)n(G)+\frac12(D-2)N(G)-D.\label{eq-supdeg}
\end{equation}
\end{theorem}
\begin{proof}
We use the expression \eqref{eq-amplattr} of $A_G^\mu$ to bound
\begin{equation*}
\caA_G^\mu=\int\dd^Dk_1\dots\dd^Dk_N\ A_G^\mu(k_1,\dots,k_N)\hat\varphi_1(k_1)\dots\hat\varphi_N(k_N)
\end{equation*}
where $\varphi_j\in\caS(\gR^D)$, i.e. $\varphi_j$ are Schwartz functions (smooth and rapidly decresing at infinity), and we recall that $k_e$ are the external impulsions. First, let us solve $\Delta$. We choose a spanning rooted\footnote{spanning means that the tree $\caT$ reaches every vertex of $G$; rooted: we fix a vertex to be the root of the tree.} tree $\caT$ of the graph $G$ and perform the following change of variables: $u_\ell=p_\ell+q_\ell$ and $v_\ell=p_\ell-q_\ell$. Then, the $\Delta$ function allows to evaluate all the variables $v_\ell$ of the lines $\ell$ of the tree $\caT$, with rest:
\begin{equation}
\delta(2\sum_{\ell\in \caI(G)}u_\ell-2\sum_e k_e).\label{eq-overalldelta}
\end{equation}
Note that this rest depends on the variables $u_\ell$ because of the translation invariance breaking of the model. With this rest, one can evaluate the variable $k_{N(G)}=\sum_\ell u_\ell-\sum_{e<N}k_e$. By using Cauchy-Schwartz inequality, we obtain
\begin{equation*}
|\int\dd^Dk_{N-1}\ \hat\varphi_{N-1}(k_{N-1})\hat\varphi_N(\sum_\ell u_\ell-\sum_{e<N-1}k_e-k_{N-1})|\leq \norm \hat\varphi_{N-1}\norm_2\norm\hat\varphi_N\norm_2.
\end{equation*}
Due also to Lemma \ref{lem-prop}, we get the bound:
\begin{multline*}
|\caA_G^\mu|\leq K\norm\hat\varphi_1\norm_1\dots\norm\hat\varphi_{N-2}\norm_1\norm\hat\varphi_{N-1}\norm_2\norm\hat\varphi_N\norm_2 \int\left(\prod_{\ell\in\caI(G)}\dd^Du_\ell\,M^{(i_\ell+1)(D-2)}e^{-kM^{2(i_\ell+1)}u_\ell^2}\right)\\
\left(\prod_{\ell\in\caI(G)\setminus\caT} \dd^Dv_\ell\,e^{-k'M^{-2(i_\ell+1)}v_\ell^2}\right)
\end{multline*}
by neglecting some exponentials. Note also that $K>0$ denotes a new constant at each inequality for simplicity (depending only on $n(G)$ and $I(G)$). Next,
\begin{equation}
|\caA_G^\mu|\leq K\norm\hat\varphi_1\norm_1\dots \norm\hat\varphi_N\norm_2 \prod_{\ell\in\caI(G)}M^{-2(i_\ell+1)}\prod_{\ell\in\caI(G)\setminus\caT} M^{D(i_\ell+1)}\label{eq-ampl1}
\end{equation}
by performing the Gaussian integrals of the remaining variables.

We use the concept of quasilocal subgraphs $G^i_k$ introduced in Definition \ref{def-quasilocal}:
\begin{equation}
\prod_{\ell\in\caI(G)}M^{\alpha(i_\ell+1)}=\prod_{\ell\in\caI(G)}\prod_{j=0}^{i_\ell}M^{\alpha}=\prod_{\ell\in\caI(G)} \,\prod_{\substack{(i,k)|\\ \ell\in\caI(G^i_k)}}M^{\alpha}=\prod_{(i,k)}\prod_{\ell\in\caI(G^i_k)}M^{\alpha}=\prod_{(i,k)} M^{\alpha I(G^i_k)}.\label{eq-index1}
\end{equation}
To optimize our estimates, we choose the spanning tree $\caT$ which solves the function $\Delta$ to be with lines $\ell$ having highest indices $i_\ell$ as possible. This means that $\caT\cap\caI(G^i_k)$ is a spanning tree of $G^i_k$ and in particular that $\sharp(\caI(G^i_k)\setminus\caT)$ is equal to the number of loops $L(G^i_k)$. Next, like for \eqref{eq-index1},
\begin{equation}
\prod_{\ell\in\caI(G)\setminus\caT}M^{\alpha(i_\ell+1)}=\prod_{(i,k)}\prod_{\ell\in\caI(G^i_k)\setminus\caT}M^{\alpha}=\prod_{(i,k)} M^{\alpha L(G^i_k)}.\label{eq-index2}
\end{equation}
Inserting \eqref{eq-index1} and \eqref{eq-index2} into \eqref{eq-ampl1} and using the identity \eqref{eq-euler}, we find
\begin{equation*}
|\caA_G^\mu|\leq K\norm\hat\varphi_1\norm_1\dots \norm\hat\varphi_N\norm_2\prod_{(i,k)}M^{-\omega(G^i_k)},
\end{equation*}
where $\omega(G^i_k)$ is an integer (due to the parity of the theory) given by \eqref{eq-supdeg}.
\end{proof}

The expression of the superficial degree of divergence \eqref{eq-supdeg} of the theory implies that:
\begin{itemize}
\item if $D<4$, there is only a finite number of graphs $G$ (with parameters $(n(G),N(G))$) such that $\omega(G)\leq 0$, so which need to be renormalized. The theory will be superrenormalizable.
\item if $D=4$, there is an infinite number of graphs $G$ with $\omega(G)\leq 0$ but always with $N(G)=2$ or $N(G)=4$. The theory will be just renormalizable.
\item if $D>4$, there is an infinite number of graphs $G$ with $\omega(G)\leq 0$ and arbitrary $N(G)$. The theory is non renormalizable.
\end{itemize}
In the rest of this subsection and in the next one, we restrict to the case of $D\leq 4$ and we will show that the theory is renormalizable.
\begin{definition}
We call a connected graph $G$ {\bf strongly convergent} if $\forall H$ connected subgraph of $G$, $\omega(H)>0$.
\end{definition}

\begin{theorem}
\label{thm-powcount2}
Let $G$ be a connected amputated graph of the theory \eqref{eq-com-actharm} which is strongly convergent. Then there exists a constant $K>0$ such that $\forall \varphi_j\in\caS(\gR^D)$,
\begin{equation*}
|\int\dd^Dk_1\dots\dd^Dk_N\ A_G(k_1,\dots,k_N)\hat\varphi_1(k_1)\dots\hat\varphi_N(k_N)|\leq K\norm\hat\varphi_1\norm_1\dots \norm\hat\varphi_{N-2}\norm_1\norm\hat\varphi_{N-1}\norm_2\norm\hat\varphi_N\norm_2.
\end{equation*}
\end{theorem}
\begin{proof}
Let $G$ be a connected strongly convergent graph of the theory. By use of Theorem \ref{thm-powcount}, we just have to sum over all the attributions $\mu$ (with $\rho\to\infty$). By a careful analysis on $\omega(G)$ for $D\in\{0,\dots,4\}$ and by using the identity $N(G)\leq 2n(G)+2$ (obtained from \eqref{eq-euler}), we can show that
\begin{equation}
\omega(G)\geq\frac14 N(G).\label{eq-supdeg2}
\end{equation}
The value $\frac14$ is the optimal one since it is reached for $(D=3,\,N=4,\,n=2)$. We then use Theorem \ref{thm-powcount} and follow the lines of the corresponding proof given in \cite{VignesTourneret:2006xa} which we reproduce here for selfcompleteness. For a vertex $\nu\in\caV(G)$, we set $e_\nu=\max_{\ell\in\caL_\nu(G)}i_\ell$ and $i_\nu=\min_{\ell\in\caL_\nu(G)}i_\ell$ to be the highest and the lowest indices of lines hooked to $\nu$. Then, for $i\in\{0,\dots,\rho\}$, we can see that $\nu\in\caV(G^i_k)$ is an external vertex\footnote{an external vertex is a vertex hooked at an external legs at least} of $G^i_k$ if and only if $i_\nu<i\leq e_\nu$. Since the number of external vertices is smaller than the number of external legs and due to \eqref{eq-supdeg2}, we have
\begin{multline*}
\prod_{(i,k)}M^{-\omega(G^i_k)}\leq \prod_{(i,k)}M^{-\frac14 N(G^i_k)}\leq \prod_{(i,k)}\prod_{\substack{\nu\in\caV(G^i_k)|\\i_\nu<i\leq e_\nu}}M^{-\frac14}=\prod_{\nu\in\caV(G)}\prod_{\substack{(i,k)|\\i_\nu<i\leq e_\nu}}M^{-\frac14}=\\
\prod_{\nu\in\caV(G)}M^{-\frac14(e_\nu-i_\nu)}
\end{multline*}
For the $\phi^4$-theory, we have $\forall\nu\in\caV(G)$, $e_\nu-i_\nu\geq\frac18\sum_{\ell,\ell'\in\caL_\nu(G)}|i_\ell-i_{\ell'}|$, which implies that
\begin{equation*}
\prod_{(i,k)}M^{-\omega(G^i_k)}\leq \prod_{\nu\in\caV(G)}\prod_{\ell,\ell'\in\caL_\nu(G)}M^{-\frac{|i_\ell-i_{\ell'}|}{32}}.
\end{equation*}
Finally, we choose another rooted tree $\caT'$ (not spanning) of $G$ such that every line of $G$ is hooked to a vertex of the tree. Therefore, for any $\ell\in\caI(G)$, there exists $\nu\in\caV(\caT')$ (we fix it) such that $\ell\in\caL_\nu(G)$. Moreover, there exists a unique line $\ell_\nu$ of the tree $\caT'$ hooked to $\nu$ and directed towards the root. We can perform the unambiguous change of variables: $j_\ell=|i_\ell-i_{\ell_\nu}|$. Then,
\begin{equation*}
\prod_{(i,k)}M^{-\omega(G^i_k)}\leq \prod_{\nu\in\caV(\caT')}\prod_{\ell,\ell'\in\caL_\nu(G)}M^{-\frac{|i_\ell-i_{\ell'}|}{32}}\leq \prod_{\ell\in\caI(G)} M^{-\frac{j_\ell}{32}}.
\end{equation*}
Allowing $\rho\to\infty$ and summing on the attributions $\mu=\{i_\ell\}_{\ell\in\caI(G)}$, we find
\begin{multline*}
|\sum_\mu \caA_G^\mu|\leq K\norm\hat\varphi_1\norm_1\dots \norm\hat\varphi_N\norm_2\sum_\mu \prod_{\ell\in\caI(G)} M^{-\frac{j_\ell}{32}}\\
\leq K\norm\hat\varphi_1\norm_1\dots \norm\hat\varphi_N\norm_2 \prod_{\ell\in\caI(G)}\sum_{j_\ell=0}^\infty M^{-\frac{j_\ell}{32}}\leq K\norm\hat\varphi_1\norm_1\dots \norm\hat\varphi_{N-2}\norm_1\norm\hat\varphi_{N-1}\norm_2\norm\hat\varphi_N\norm_2.
\end{multline*}
\end{proof}

\subsection{Renormalization}

In the above subsection, we have seen that a graph $G$ endowed with an attribution $\mu$ has a convergent amplitude if all its quasilocal subgraphs $G^i_k$ have strictly positive superficial degree of divergence $\omega(G^i_k)$ (see Theorem \ref{thm-powcount}). Then, for a general graph $(G,\mu)$, the subgraphs to renormalize are the quasilocal ones $G^i_k$ such that $\omega(G^i_k)\leq 0$ (and with at least one loop).
\begin{itemize}
\item In dimension $D=4$, they correspond to quasilocal subgraphs $G_k^i$ with $N(G^i_k)=2,4$ external legs.
\item In dimension $D=3$, they correspond to quasilocal subgraphs $G_k^i$ with $N(G^i_k)=2$ external legs and $n(G^i_k)=1,2$. There are only 4 (topologically different) such subgraphs.
\item In dimension $D=2$, they correspond to quasilocal subgraphs $G_k^i$ with $N(G^i_k)=2$ external legs and $n(G^i_k)=1$, which in fact correspond to only one subgraph called the ``tadpole''.
\item In dimension $D=0,1$, there is no subgraph to renormalize.
\end{itemize}
In the proof of Theorem \ref{thm-powcount}, we have seen that the amplitude $A_G$ \eqref{eq-ampl} of a graph $G$ contains the overall delta function \eqref{eq-overalldelta}, so that we define $B_G$ as
\begin{equation*}
A_G(k_1,\dots,k_{N-1},k_N)=\int \prod_{\ell\in\caI(H)}\dd^Du_\ell\ B_G(k_1,\dots,k_{N-1},k_N,u_\ell)\delta\Big(\sum_{j=1}^Nk_j-\sum_{\ell\in\caI(H)}u_\ell\Big).
\end{equation*}
Then we note $\tau$ the {\bf Taylor operator} acting on amplitudes: if $\omega(G)\leq 0$,
\begin{equation*}
\tau A_G(k_1,\dots,k_N):=\sum_{j=0}^{-\omega(G)}\frac{1}{j!}\frac{\dd^j}{\dd t^j}\int \prod_{\ell\in\caI(H)}\dd^Du_\ell\ B_G(tk_1,\dots,tk_{N-1},k_N,u_\ell)\delta\Big(\sum_{j=1}^Nk_j-t\sum_{\ell\in\caI(H)}u_\ell\Big)|_{t=0}.
\end{equation*}
The action of $\tau$ on the amplitudes will allow to isolate the divergences and to renormalize them.

\begin{theorem}
\label{thm-renorm}
Let $(G,\mu)$ be a connected graph of the theory \eqref{eq-com-actharm} with an attribution. For any quasilocal subgraph $H=G^i_k$ with $\omega(H)\leq 0$, 
\begin{itemize}
\item The counterterm $\tau\caA_{H}^\mu(\phi):=\int\dd^Dk_1\dots\dd^Dk_N\ \tau A_{H}^\mu(k_1,\dots,k_N)\hat\phi(k_1)\dots\hat\phi(k_N)$ is a functional of the form of \eqref{eq-com-actharm} in the field $\phi$ (but divergent for $i\to\infty$ together with $\rho\to\infty$).
\item The rest $(1-\tau)\caA_{H}^\mu(\phi)$ is convergent for $i\to\infty$.
\end{itemize}
\end{theorem}
\begin{proof}
\begin{itemize}
\item We have seen that only quasilocal subgraphs with $N=2$ or $N=4$ external legs have negative superficial degree of divergence. Let us start with the case of $N=4$ (only needed in $D=4$ dimensions). We consider $(G,\mu)$ a connected graph with an attribution and $H=G^i_k$ a quasilocal subgraph with $N=4$ external legs so that $\omega(H)=0$. In the notations of the above subsection,
\begin{equation*}
\caA_H^\mu(\phi)=\int \dd^Dk_1\dots\dd^Dk_N\,\hat\phi(k_1)\dots\hat\phi(k_N)\prod_{\ell\in\caI(H)}\dd^Dp_\ell\dd^Dq_\ell\ \hat C^{i_\ell}(p_\ell,q_\ell)\ \Delta(\{p_\ell,q_\ell,k_j\}).
\end{equation*}
By the proof of Theorem \ref{thm-powcount}, we can choose a spanning tree $\caT$ of $H$, introduce variables $u_\ell,v_\ell$, and evaluate the $\Delta$-function: for any line of the tree $\ell\in\caI(\caT)$, the variable $v_\ell$ can be written in terms of the other variables:
\begin{equation*}
v_\ell=V_\ell+U_\ell+K_\ell,\text{ where }V_\ell=\sum_{\ell'\in\caI(H)\setminus\caT}\alpha_{\ell,\ell'} v_{\ell'},\quad U_\ell=\sum_{\ell'\in\caI(H)}\beta_{\ell,\ell'} u_{\ell'},\quad K_\ell=\sum_{j=1}^N\gamma_{\ell,j} k_{j},
\end{equation*}
where $\alpha,\beta,\gamma$ are some coefficients. Moreover, the overall delta function \eqref{eq-overalldelta} remains. It gives:
\begin{multline*}
\caA_H^\mu(\phi)=K\int \prod_{j=1}^N\dd^Dk_j\,\hat\phi(k_j)\prod_{\ell\in\caI(H)\setminus\caT}\dd^Du_\ell\dd^Dv_\ell\ \tilde C^{i_\ell}(u_\ell,v_\ell) \prod_{\ell\in\caT}\dd^Du_\ell\ \tilde C^{i_\ell}(u_\ell,V_\ell+U_\ell+K_\ell)\\
\delta(\sum_{j=0}^N k_j-\sum_{\ell\in \caI(H)}u_\ell),
\end{multline*}
where $K$ is a constant and we write for simplicity $\tilde C^{i_\ell}(u_\ell,v_\ell)=\hat C^{i_\ell}(\frac12(u_\ell+v_\ell),\frac12(u_\ell-v_\ell))$. It can be reexpressed as
\begin{multline}
\caA_H^\mu(\phi)=K\int \prod_{j=1}^N\dd^Dk_j\,\hat\phi(k_j)\prod_{\ell\in\caI(H)\setminus\caT}\dd^Du_\ell\dd^Dv_\ell\ \tilde C^{i_\ell}(u_\ell,v_\ell)\\
\prod_{\ell\in\caT}\dd^Du_\ell\ \tilde C^{i_\ell}(u_\ell,V_\ell+U_\ell+sK_\ell)
\delta(\sum_{j=0}^N k_j-s\sum_{\ell\in \caI(H)}u_\ell)|_{s=1}.\label{eq-expampl}
\end{multline}
We use now the Taylor expansion on the variable $s$: $f(1)=f(0)+\int_0^1\frac{\partial}{\partial s}f(s)\dd s$. For $s=0$, we recognize exactly the expression $\tau\caA_H^\mu(\phi)$, which is just the amplitude taken for vanishing external impulsions since $\omega(H)=0$. So, the counterterm has the form
\begin{equation*}
\tau\caA_H^\mu(\phi)=K(\mu)\int \prod_{j=1}^4\dd^Dk_j\,\hat\phi(k_j)\delta(k_1+\dots+k_4)=K'(\mu)\int\dd^Dx\ \phi(x)^4,
\end{equation*}
while the rest is given by
\begin{multline*}
(1-\tau)\caA_H^\mu(\phi)=K\int \prod_{j=1}^4\dd^Dk_j\,\hat\phi(k_j)\prod_{\ell\in\caI(H)\setminus\caT}\dd^Du_\ell\dd^Dv_\ell\ \tilde C^{i_\ell}(u_\ell,v_\ell)\int_0^1\dd s \prod_{\ell\in\caT}\dd^Du_\ell\\
 \tilde C^{i_\ell}(u_\ell,V_\ell+U_\ell+sK_\ell)
\left(-U_\nu\partial_\nu\delta\Big(\sum_{j=0}^4 k_j-sU\Big)+D(u_\ell,U_\ell+V_\ell,K_\ell,s)\delta\Big(\sum_{j=0}^4 k_j-sU\Big)\right).
\end{multline*}
where $U=\sum_{\ell\in \caI(H)}u_\ell$ and $\frac{\partial}{\partial s}\tilde C^{i_\ell}(u_\ell,U_\ell+V_\ell+sK_\ell)= D(u_\ell,U_\ell+V_\ell,K_\ell,s) \tilde C^{i_\ell}(u_\ell,U_\ell+V_\ell+sK_\ell)$. By an analysis similar to the one leading to Equation \eqref{eq-ampl1}, we can observe that the term in $U_\nu\partial_\nu\delta$ brings at least $M^{-\omega(H)-i}=M^{-i}$ and norms of the type $\norm\hat\phi\norm_1$, $\norm\hat\phi\norm_2$, $\norm\partial_\nu \hat\phi\norm_2$, $\norm k_\nu\hat\phi\norm_2$. The second term, in $D(u_\ell,U_\ell+V_\ell,K_\ell,s)\delta$ brings also at least $M^{-i}$ and other norms like $\norm k_\nu k_\sigma\hat\phi\norm_2$. This concludes the second assertion of the Theorem for the case $N=4$ since $M^{-i}\to 0$.

\item We focus now on the case of $N=2$ external legs with $\omega\leq 0$, in dimension $D=2,3,4$, whose amplitude is given by \eqref{eq-expampl}. As before, we perform a Taylor expansion on the variable $s$ but at the third order: $f(1)=f(0)+f'(0)+\frac12f''(0)+\frac12\int_0^1(1-s)^2f^{(3)}(0)\dd s$. The first three terms of the expansion coincide with $\tau \caA_H^\mu(\phi)$ which takes the form
\begin{equation*}
\tau \caA_H^\mu(\phi)=K\int\dd^Dk\ \hat\phi(k)\Big(\alpha_1(\mu)+\alpha_2(\mu)k^2+\alpha_3(\mu)\partial_k^2\Big)\hat\phi(-k)
\end{equation*}
after calculations, or by the Parseval-Plancherel theorem,
\begin{equation*}
\tau \caA_H^\mu(\phi)=K'\int\dd^Dx\ \phi(x)\Big(\alpha_1'(\mu)+\alpha_2'(\mu)\partial_x^2+\alpha_3'(\mu)x^2\Big)\phi(x).
\end{equation*}
The fourth term, the rest $(1-\tau)\caA_H^\mu(\phi)$ is convergent when $i\to\infty$, by the same analysis as before.
\end{itemize}
\end{proof}

Due to Theorem \ref{thm-renorm} and Theorem \ref{thm-powcount2}, BPHZ renormalization can be directly performed by defining recursively the renormalized amplitudes, e.g.
\begin{equation*}
\caA_H^R:=(1-\tau)\caA_H(\phi)
\end{equation*}
if $H$ does not contain subgraphs with negative degree of divergence $\omega$; and by using the Zimmermann forest formula \cite{Zimmermann:1969}. See also chapter 1 of \cite{VignesTourneret:2006xa} for a good overview on this procedure. Thus, we have shown here that the theory \eqref{eq-com-actharm} is renormalizable at all orders in perturbation.

\section{Beta functions}
\label{sec-beta}

\subsection{Commutative theory}

In this subsection, we perform the calculation of the one-loop beta function of the model \eqref{eq-com-actharm} in $D=4$ dimensions. Due to the Feynman rules \eqref{eq-ampl}, we find the one-loop two point correlation function as:
\begin{equation}
\caA_2(x)=-12 \lambda C(x,x)\phi^2(x)=\frac{-3\lambda\Omega^2}{4\pi^2}\int_\eps^1\frac{\dd t}{4\Omega^2 t^2}(1-m^2t-\Omega^2t x^2)\phi^2(x),\label{eq-com-ampl2}
\end{equation}
where the integral on $t$ is regularized by a UV cut-off $\eps$ near 0, and up to finite contributions in $\eps\to 0$. It gives:
\begin{equation*}
\caA_2(x)=\frac{-3\lambda}{16\pi^2}\Big(\frac{1}{\eps}+m^2\ln(\eps)+\Omega^2 x^2\ln(\eps)\Big)\phi^2(x).
\end{equation*}

In the same way, the one-loop four point correlation function is:
\begin{equation*}
\caA_4(x)=288\lambda^2\int\dd y\ \phi(x)^2\phi(y)^2 C(x,y)^2
\end{equation*}
Therefore, with a Taylor expansion on the fields, it becomes:
\begin{equation*}
\caA_4(x)=288\lambda^2 \int\dd z\ \phi(x)^2 C(x,x+z)^2 (\phi(x)+z_\mu\partial_\mu\phi(x)+\frac12 z_\mu z_\nu\partial_\mu \partial_\nu\phi(x)+\dots)^2
\end{equation*}
Using the expression of the propagator \eqref{eq-propag3} regularized as above by a cut-off $\eps$, we obtain:
\begin{equation*}
\caA_4(x)=\frac{-9\lambda^2}{8\pi^2}\ln(\eps)\phi^4(x),
\end{equation*}
up to finite contributions.

Combining these contributions, we find that the one-loop effective action can be expressed as:
\begin{equation}
\Gamma_{1l}(\phi)= \int \dd^4x\Big(\frac 12(\partial_\mu\phi)^2 +\frac{\Omega^2}{2}(1+\frac{3\lambda\ln(\eps)}{8\pi^2})x^2\phi^2 +\frac{m^2}{2}(1+\frac{3\lambda}{8\pi^2 m^2\eps}+\frac{3\lambda\ln(\eps)}{8\pi^2})\phi^2 +\lambda(1+\frac{9\lambda\ln(\eps)}{8\pi^2})\phi^4\Big).
\label{eq-com-acteff}
\end{equation}
We express the one-loop beta functions with respect to the physical (renormalized) constants: $\lambda_R$, $m^2_R$, $\Omega_R$. Here the renormalization of the wave function is $Z=1$ because of the coefficient of $(\partial_\mu \phi)^2$ in \eqref{eq-com-acteff}. The gamma function (for the renormalization of the wave function) and the {\bf beta functions} of the different constants are:
\begin{align}
&\beta_\lambda :=\frac{\partial \lambda}{\partial(-\ln(\eps))}=\frac{9\lambda_R^2}{8\pi^2},\qquad \beta_\Omega :=\frac{\partial \Omega}{\partial(-\ln(\eps))}=\frac{3\lambda_R \Omega_R}{16\pi^2},\nonumber\\
&\beta_{m^2}:=\frac{1}{m^2}\frac{\partial m^2}{\partial(-\ln(\eps))}=\frac{3\lambda_R}{8\pi^2}-\frac{3\lambda_R}{8\pi^2 m^2_R\eps},\qquad \gamma:=\frac{\partial \ln(Z)}{\partial(-\ln(\eps))}=0. \label{eq-com-beta}
\end{align}

\subsection{Theory on the Moyal space}
\label{subsec-moyal}

We present here the beta function of the scalar field theory on the Moyal space with harmonic term given by
\begin{equation}
S(\phi)=\int \dd^4x\Big(\frac 12(\partial_\mu\phi)^2 +\frac{\Omega^2}{2}\wx^2\phi^2 +\frac{m^2}{2}\phi^2 +\lambda\phi\star\phi\star\phi\star\phi\Big),\label{eq-moy-actharm}
\end{equation}
where $\Theta$ is a non-degenerate skewsymmetric matrix, $\wx=2\Theta^{-1}x$, $\Omega$ is a real dimensionless parameter, and the Moyal product is given by
\begin{equation*}
(f\star g)(x)=\frac{1}{\pi^4\theta^4}\int \dd^4 y\dd^4 z\ f(y)g(z) e^{-iy\widetilde z-iz\widetilde x-ix\widetilde y}.
\end{equation*}
This theory is renormalizable to all orders \cite{Grosse:2004yu} for a non-zero $\Omega$. Note that this action has been interpreted as coming from a deformation of a superspace \cite{Bieliavsky:2010su,deGoursac:2011kv}. Its propagator is given by the Mehler kernel in the position space \cite{Gurau:2005qm}:
\begin{align}
C(x,y)&=\left(\frac{\Omega}{\pi\theta}\right)^{2}\int_0^\infty \!\! \frac{\dd t}{\sinh^{2}(\frac{4\Omega t}{\theta})} e^{-m^2t}C(t,x,y),\nonumber\\
C(t,x,y) &= \exp\Big(-\frac{\Omega}{2\theta}\coth(\frac{2\Omega t}{\theta})(x-y)^2 -\frac{\Omega}{2\theta}\tanh(\frac{2\Omega t}{\theta})(x+y)^2\Big)\label{eq-propag1}
\end{align}
which corresponds to \eqref{eq-propag3} up to a redefinition of the constant $\Omega$.

Owing to the computations of \cite{Grosse:2004by}, the one-loop effective action of the theory \eqref{eq-moy-actharm} can then be showed to be:
\begin{multline*}
\Gamma_{1l}(\phi)= \int \dd^4x\Big(\frac 12(1+\frac{\lambda\Omega^2\ln(\eps)}{4\pi^2(1+\Omega^2)^3})(\partial_\mu\phi)^2 +\frac{\Omega^2}{2}(1+\frac{\lambda\ln(\eps)}{4\pi^2(1+\Omega^2)^3})\wx^2\phi^2\\
+\frac{m^2}{2}(1+\frac{\lambda}{4\pi^2 m^2(1+\Omega^2)^2\eps}+\frac{\lambda\ln(\eps)}{4\pi^2(1+\Omega^2)^2})\phi^2 +\lambda(1+\frac{\lambda\ln(\eps)}{4\pi^2(1+\Omega^2)^2})\phi^{\star 4}\Big)
\end{multline*}
in the regularization scheme of expression \eqref{eq-com-ampl2}. Note that it does not correspond to the one chosen in \cite{Grosse:2004by}. We present here the results in this scheme in order to compare them with the beta functions of the commutative case and of the Moyal degenerate case, where a similar scheme is chosen. Due to the non-trivial contribution of the kinematic term $(\partial_\mu\phi)^2$, we have to perform the change of field $\phi=Z^{\frac12}\phi_R$. Since the effective action takes the form
\begin{equation*}
\Gamma_{1l}(\phi)= \int \dd^4x\Big(\frac 12(\partial_\mu\phi_R)^2 +\frac{\Omega_R^2}{2}\wx^2\phi_R^2+\frac{m_R^2}{2}\phi_R^2 +\lambda_R\phi_R^{\star 4}\Big),
\end{equation*}
in function of the physical (renormalized) constants: $\lambda_R$, $m^2_R$, $\Omega_R$, we deduce that the renormalization of the wave function is:
\begin{equation*}
Z=(1+\frac{\lambda\Omega^2\ln(\eps)}{4\pi^2(1+\Omega^2)^3})^{-1}.
\end{equation*}
By expanding the bare constants $\lambda$, $m^2$, $\Omega$ and $Z$ in terms of $\lambda_R$, we can compute the following beta functions:

\begin{align}
&\beta_\lambda=\frac{\lambda_R^2(1-\Omega_R^2)}{4\pi^2(1+\Omega_R^2)^3},\qquad
\beta_\Omega =\frac{\lambda_R \Omega_R(1-\Omega_R^2)}{8\pi^2(1+\Omega_R^2)^3},\nonumber\\
&\beta_{m^2}=\frac{\lambda_R}{4\pi^2(1+\Omega_R^2)^3}-\frac{\lambda_R}{4\pi^2 m^2_R(1+\Omega_R^2)^2\eps},\qquad\gamma=\frac{\lambda_R\Omega_R^2}{4\pi^2(1+\Omega_R^2)^3}.\label{eq-moy-beta}
\end{align}
We see that for $\Omega=1$, $\beta_\lambda=\beta_\Omega=0$. These results coincide with \cite{Grosse:2004by} up to the change of regularization scheme.

\subsection{Theory on the degenerate Moyal space}

Let now $\Theta$ be a skewsymmetric matrix degenerate in the two first coordinates. Each position $x\in \gR^4$ will be written as $(x_0,x_*)$ where $x_0=(x_1,x_2)\in\gR^2$ correspond to the two first coordinates, while $x_*=(x_3,x_4)\in\gR^2$ represent the third and the fourth ones. The associated star-product is given by:
\begin{equation*}
(f\star g)(x)=\frac{1}{\pi^2\theta^2}\int \dd y_*\dd z_*\ f(x_0,y_*)g(x_0,z_*) e^{-iy_*\wedge z_*-iz_*\wedge x_*-ix_*\wedge y_*}
\end{equation*}
where $y_*\wedge z_*=2y_* \Theta_*^{-1} z_*$ and $\Theta_*$ represents the non-degenerate part of $\Theta$. The action we want to consider here is
\begin{multline}
S(\phi)=\int \dd^4x\Big(\frac 12(\partial_\mu\phi)^2 +\frac{\Omega_0^2}{2}x_0^2\phi^2+\frac{\Omega_*^2}{2}\wx_*^2\phi^2 +\frac{m^2}{2}\phi^2 +\lambda\phi\star\phi\star\phi\star\phi\Big)\\
+\frac{\kappa^2}{\theta^2}\int \dd^2 x_0\dd^2 y_*\dd z_*\ \phi(x_0,y_*)\phi(x_0,z_*),\label{eq-deg-actharm}
\end{multline}
where $\wx_*=2\Theta_*^{-1}x_*$, $\Omega_0$ and $\Omega_*$ are respectively a dimensionful and a dimensionless parameter. The term with parameter $\kappa$ has been added for the renormalizability of the theory \cite{Grosse:2008df}. The corresponding propagator is
\begin{align}
C(x,y)&=\frac{\Omega_0\Omega_*}{8\pi^2\theta}\int_0^\infty \!\! \frac{\dd t}{\sinh(2\Omega_0 t)\sinh(\frac{4\Omega_* t}{\theta})} e^{-m^2t}C(t,x,y),\nonumber\\
C(t,x,y) &= \exp\Big(-\frac{\Omega_0}{4}\coth(\Omega_0 t)(x_0-y_0)^2 -\frac{\Omega_0}{4}\tanh(\Omega_0 t)(x_0+y_0)^2\nonumber\\
&\quad-\frac{\Omega_*}{2\theta} \coth(\frac{2\Omega_* t}{\theta})(x_*-y_*)^2 -\frac{\Omega_*}{2\theta}\tanh(\frac{2\Omega_* t}{\theta})(x_*+y_*)^2\Big).\label{eq-propag2}
\end{align}

In \cite{Grosse:2008df} it has been observed that only terms for $\kappa=0$ were involved in the renormalization of the wave function, the harmonic term, the mass term and the quartic interaction term. So, we compute the corresponding part of the one-loop effective action by setting $\kappa=0$. At the level of the planar regular part of the one-loop two point correlation function, it gives
\begin{equation*}
\caA_2(x)=\frac{-\lambda}{8\pi^2(1+\Omega_*^2)}\Big(\big(\frac{1}{\eps}+m^2\ln(\eps)+\Omega_0^2 \wx_0^2\ln(\eps) +\frac{\Omega_*^2}{1+\Omega_*^2} \wx_*^2\ln(\eps)\big)\phi^2(x)-\frac{\Omega_*^2}{1+\Omega_*^2}\ln(\eps)\phi(x)\partial^2_*\phi(x)\Big),
\end{equation*}
where $\partial^2_*=\frac{\partial^2}{\partial x_3^2}+\frac{\partial^2}{\partial x_4^2}$ and $\partial^2_0=\frac{\partial^2}{\partial x_1^2}+\frac{\partial^2}{\partial x_2^2}$.
The two-point graph with one loop and two broken faces gives also a divergent contribution, contrary to the fully noncommutative case, but it contributes only to the renormalization of the constant $\kappa$ \cite{Grosse:2008df}, so we do not compute it here. Then, the computation of the planar regular one-loop four point correlation function yields
\begin{equation*}
\caA_4(x)=\frac{-\lambda^2}{4\pi^2(1+\Omega_*^2)}\ln(\eps)\phi^{\star 4}(x),
\end{equation*}
up to finite contributions. The one-loop effective action can therefore be expressed as
\begin{multline*}
\Gamma_{1l}(\phi)= \int \dd^4x\Big(\frac{1}{2}(\partial_0\phi)^2+\frac 12(1+\frac{\lambda\Omega_*^2\ln(\eps)}{4\pi^2(1+\Omega_*^2)^2})(\partial_*\phi)^2 +\frac{\Omega_0^2}{2}(1+\frac{\lambda\ln(\eps)}{4\pi^2(1+\Omega_*^2)})\wx_0^2\phi^2\\
+\frac{\Omega_*^2}{2}(1+\frac{\lambda\ln(\eps)}{4\pi^2(1+\Omega_*^2)^2})\wx_*^2\phi^2
+\frac{m^2}{2}(1+\frac{\lambda}{4\pi^2 m^2(1+\Omega_*^2)\eps}+\frac{\lambda\ln(\eps)}{4\pi^2(1+\Omega_*^2)})\phi^2 +\lambda(1+\frac{\lambda\ln(\eps)}{4\pi^2(1+\Omega_*^2)})\phi^{\star 4}\Big).
\end{multline*}
We see that the Laplacian terms renormalize differently for the commutative coordinates and the noncommutative ones. We decide to introduce a coefficient $a$ in front of the commutative part of the Laplacian which has to be renormalized. Then we can define the wave function renormalization. But note that this procedure is not unique: we could have chosen to introduce a coefficient in front of the noncommutative terms. This problem cannot be avoided for the degenerate Moyal case.

We perform the change of fields $\phi=Z^{\frac12}\phi_R$. Since the effective action takes the form
\begin{equation*}
\Gamma_{1l}(\phi)= \int \dd^4x\Big(\frac{a}{2}(\partial_0\phi_R)^2 +\frac 12(\partial_*\phi_R)^2 +\frac{\Omega_{0,R}^2}{2}\wx_0^2\phi_R^2+\frac{\Omega_{*,R}^2}{2}\wx_*^2\phi_R^2+\frac{m_R^2}{2}\phi_R^2 +\lambda_R\phi_R^{\star 4}\Big),
\end{equation*}
in function of the physical (renormalized) constants: $\lambda_R$, $m^2_R$, $\Omega_{0,R}$, $\Omega_{*,R}$, the renormalization of the wave function is given by:
\begin{equation*}
Z=(1+\frac{\lambda\Omega_*^2\ln(\eps)}{4\pi^2(1+\Omega_*^2)^2})^{-1}.
\end{equation*}
We expand as before the bare constants $\lambda$, $m^2$, $\Omega_0$, $\Omega_*$ and $Z$ in terms of $\lambda_R$, and we obtain the beta functions:
\begin{align}
&\beta_\lambda =\frac{\lambda_R^2(1-\Omega_{*,R}^2)}{4\pi^2(1+\Omega_{*,R}^2)^2},\qquad
\beta_{\Omega_0} =\frac{\lambda_R \Omega_{0,R}}{8\pi^2(1+\Omega_{*,R}^2)^2},\qquad
\beta_{\Omega_*} =\frac{\lambda_R \Omega_{*,R}(1-\Omega_{*,R}^2)}{8\pi^2(1+\Omega_{*,R}^2)^2},\nonumber\\
&\beta_{m^2}=\frac{\lambda_R}{4\pi^2(1+\Omega_{*,R}^2)^2}-\frac{\lambda_R}{4\pi^2 m^2_R(1+\Omega_{*,R}^2)\eps},\qquad
\gamma=\frac{\lambda_R\Omega_{*,R}^2}{4\pi^2(1+\Omega_{*,R}^2)^2}.\label{eq-degen-beta}
\end{align}
These results have also been obtained in \cite{Grosse:2012my} for the case $\Omega_0=0$. But here, one need to treat also the commutative directions with a harmonic potential for the discussion, so to consider $\Omega_0\neq 0$.

\section{Discussion}
\label{sec-disc}

We have seen in section \ref{sec-ren} that the commutative scalar field theory with harmonic term \eqref{eq-com-actharm} has the same power counting \eqref{eq-supdeg} as the usual theory without harmonic term (Theorem \ref{thm-powcount}). One could have said that the mass dimension of the parameter $\Omega$ straighforwardly gave this result, but the usual argument based on mass dimension of the parameters is not valid if the quadratic part of the action (here \eqref{eq-com-actharm}) involves terms with both momenta $p^2$ and positions $x^2$. That is why it was important to check this power-counting. Moreover, Theorem \ref{thm-renorm} showed the quantum stability of the theory \eqref{eq-com-actharm}, i.e. that the harmonic term does not generate other terms than the one involved in the classical action. So the theory is renormalizable to all orders in $D=4$ dimensions, and superrenormalizable for $D\leq 3$ dimensions. In particular, this result stresses the interest of the analysis of \cite{Wulkenhaar:2009pv}, describing a modified Higgs mechanism for this theory from non-constant vacuum configurations. So far, it would be interesting to continue this analysis with respect to the results of the present paper, namely to investigate the quantum stability of the vacuum configurations exhibited in \cite{Wulkenhaar:2009pv}.
\medskip

The theory \eqref{eq-moy-actharm} with harmonic term on the Moyal space is known to have a ill-defined commutative limit $\theta\to 0$, for a fixed parameter $\Omega$. The simplest solution was then to assume that the harmonic term also vanishes $\Omega\to0$ when taking the commutative limit. In the present paper we can analyze what happens if $\frac{\Omega}{\theta}$ is fixed to a certain value $\Omega_{\text{com}}>0$ during the commutative limit. Is it well defined, $\Omega$ and $\Omega_{\text{com}}$ have to run in a similar way.
\begin{itemize}
\item In the commutative case \eqref{eq-com-beta}, $\Omega_{\text{com}}$ does not have any effect on the flow of the other constants. The renormalization of the wave function does not take place at one-loop like without the harmonic term.
\item Contrary to this, in the Moyal case \eqref{eq-moy-beta}, the parameter $\Omega$ has a strong effect on the flow of the coupling constant $\lambda$ due to the non-vanishing renormalization of the wave function. $\Omega=1$ is a fixed point of the renormalization group and for this value, $\beta_\lambda=0$ \cite{Grosse:2004by,Disertori:2006nq}, so that the theory is asymptotically safe.
\end{itemize}
We see therefore that this commutative limit, preserving $\frac{\Omega}{\theta}=\Omega_{\text{com}}$, is not compatible with the renormalization flow. In particular, the commutative scalar field theory is not asymptotically safe. Notice that the property of asymptotic safety appears as soon as there are noncommutative directions, as shown by the degenerate Moyal case \eqref{eq-degen-beta}. The parameter of selfduality is then $\Omega_*$ of these noncommutative directions. Another consequence of this analysis is that the kinetic and harmonic terms are renormalized in a different way for the commutative and noncommutative directions.

One of the advantage of considering the scalar theory with harmonic term on the Moyal space is therefore this property of asymptotic safety, which could permit to define it at a constructive level. The vacuum configurations for negative mass term ($m^2<0$) have been investigated in \cite{deGoursac:2007uv}. These non-constant solutions may share some interesting features with the one examined in \cite{Wulkenhaar:2009pv}. In view of a modified Higgs mechanism, one could also look at the quantum stability of the vacuum solutions of \cite{deGoursac:2007uv} as well as at their asymptotic properties.

\vskip 1 true cm

{\bf Acknowledgements}: The author thanks Vincent Rivasseau and Fabien Vignes-Tourneret for interesting discussions on multiscale analysis and on this work. He also acknowledges the Max Planck Institut f\"ur Mathematik (Bonn) for its invitation.

\bibliographystyle{utcaps}
\bibliography{biblio-these,biblio-perso,biblio-recents}

\end{document}